\documentclass[final,twoside,11pt]{entics}

\usepackage{graphicx}
\usepackage[all]{xy}
\usepackage{main}
\usepackage{local-abbrv}

\usepackage{todonotes}
\usepackage{enticsmacro}
\def\conf{MFPS 2024}
\def\myconf{mfps40} 	
\volume{4}			
\def\volu{4}			
\def\papno{10}			
\def\mydoi{14797}%

\def\lastname{Grodin and Harper}
\def\copynames{H. Grodin, R. Harper}
\def\runtitle{Amortized Analysis via Coalgebra}

\begin{document}

\begin{frontmatter}
  \title{Amortized Analysis via Coalgebra}

  \author{Harrison Grodin\thanksref{hgrodin}}
  \author{Robert Harper\thanksref{rwh}}

  \address{Computer Science Department\\ Carnegie Mellon University\\ Pittsburgh, PA, USA}

  \thanks[hgrodin]{Email: \href{mailto:hgrodin@cs.cmu.edu}{\texttt{\normalshape hgrodin@cs.cmu.edu}}}
  \thanks[rwh]{Email: \href{mailto:rwh@cs.cmu.edu}{\texttt{\normalshape rwh@cs.cmu.edu}}}

  \begin{abstract}
    Amortized analysis is a cost analysis technique for data structures in which cost is studied in aggregate: rather than considering the maximum cost of a single operation, one bounds the total cost encountered throughout a session.
    Traditionally, amortized analysis has been phrased inductively, quantifying over finite sequences of operations.
    Connecting to prior work on coalgebraic semantics for data structures, we develop the alternative perspective that amortized analysis is naturally viewed coalgebraically in a category of cost algebras, where a morphism of coalgebras serves as a first-class generalization of potential function suitable for integrating cost and behavior.
    Using this simple definition, we consider amortization of other sample effects, non-commutative printing and randomization.
    To support imprecise amortized upper bounds, we adapt our discussion to the bicategorical setting, where a potential function is a colax morphism of coalgebras.
    We support algebraic and coalgebraic operations simultaneously by using coalgebras for an endoprofunctor instead of an endofunctor, combining potential using a monoidal structure on the underlying category.
    Finally, we compose amortization arguments in the indexed category of coalgebras to implement one amortized data structure in terms of others.
  \end{abstract}

\end{frontmatter}

\section{Introduction}\label{sec:introduction}

In computer science, it is common to prove the cost of data structure operations, guaranteeing some exact or upper bound on the amount of abstract cost incurred.
In simple cases, a tight bound can be proved about each operation in isolation.
However, an upper bound can sometimes be too loose to be insightful, because any sequential use of the data structure can only reach the worst case infrequently.
For example, an operation that uses $\cost{8}$ of cost every eight invocations and no cost otherwise can be upper bounded by $\cost{8}$, but this gives the grossly misleading perspective that a sequence of eight operations could cost up to $\cost{64}$, even though the $\cost{8}$ will only be charged once in the sequence.
Instead, the cost of an operation should be considered in conjunction with the cost of the operations that came before or may come afterwards.

To address this problem, Tarjan~\cite{tarjan>1985} developed \emph{amortized analysis}, a technique for bounding the total cost of a sequence of operations.
Rather than claiming that the cost of the aforementioned operation is upper bounded by $\cost{8}$, one can pretend that each invocation costs only $\cost{1}$, averaging out the $\cost{8}$ over the eight invocations.
While this cost bound is not precisely true, it looks approximately true from the viewpoint of a client: a sequence of eight operations costs $\cost{8}$, exactly as the $\cost{1}$-per-operation abstraction suggests it should.
We can visualize this reasoning using the following commutative diagram:
\[\begin{tikzcd}
  \bullet & \bullet & \bullet & \bullet & \bullet & \bullet & \bullet & \bullet & \bullet \\
  \bullet & \bullet & \bullet & \bullet & \bullet & \bullet & \bullet & \bullet & \bullet
  \arrow["\cost{0}", maps to, from=1-1, to=1-2]
  \arrow["\cost{0}", maps to, from=1-2, to=1-3]
  \arrow["\cost{0}", maps to, from=1-3, to=1-4]
  \arrow["\cost{0}", maps to, from=1-4, to=1-5]
  \arrow["\cost{0}", maps to, from=1-5, to=1-6]
  \arrow["\cost{0}", maps to, from=1-6, to=1-7]
  \arrow["\cost{0}", maps to, from=1-7, to=1-8]
  \arrow["\cost{8}", maps to, from=1-8, to=1-9]
  \arrow["\cost{1}", maps to, from=2-1, to=2-2]
  \arrow["\cost{1}", maps to, from=2-2, to=2-3]
  \arrow["\cost{1}", maps to, from=2-3, to=2-4]
  \arrow["\cost{1}", maps to, from=2-4, to=2-5]
  \arrow["\cost{1}", maps to, from=2-5, to=2-6]
  \arrow["\cost{1}", maps to, from=2-6, to=2-7]
  \arrow["\cost{1}", maps to, from=2-7, to=2-8]
  \arrow["\cost{1}", maps to, from=2-8, to=2-9]
  \arrow["\cost{0}", maps to, color={rgb,255:red,127;green,127;blue,127}, from=1-1, to=2-1]
  \arrow["\cost{0}", maps to, color={rgb,255:red,127;green,127;blue,127}, from=1-9, to=2-9]
\end{tikzcd}\]
Each mapping $\bullet \mapsto \bullet$ represents the cost of an operation.
The arrows along the top represent the true cost, with the first seven operations taking no cost and the eighth operation taking $\cost{8}$ of cost.
The arrows along the bottom represent the imagined cost, with all operations taking $\cost{1}$ cost.
The lightened vertical arrows, annotated with $\cost{0}$, connect the two perspectives, allowing us to state that both routes are equivalent.
To compose this commutative cost diagram out of smaller squares for each individual operation, vertical arrows cannot always be annotated with $\cost{0}$, as this would not satisfy the arithmetic constraints.
Instead:
\[\begin{tikzcd}
  \bullet & \bullet & \bullet & \bullet & \bullet & \bullet & \bullet & \bullet & \bullet \\
  \bullet & \bullet & \bullet & \bullet & \bullet & \bullet & \bullet & \bullet & \bullet
  \arrow["\cost{0}", maps to, from=1-1, to=1-2]
  \arrow["\cost{0}", maps to, from=1-2, to=1-3]
  \arrow["\cost{0}", maps to, from=1-3, to=1-4]
  \arrow["\cost{0}", maps to, from=1-4, to=1-5]
  \arrow["\cost{0}", maps to, from=1-5, to=1-6]
  \arrow["\cost{0}", maps to, from=1-6, to=1-7]
  \arrow["\cost{0}", maps to, from=1-7, to=1-8]
  \arrow["\cost{8}", maps to, from=1-8, to=1-9]
  \arrow["\cost{1}", maps to, from=2-1, to=2-2]
  \arrow["\cost{1}", maps to, from=2-2, to=2-3]
  \arrow["\cost{1}", maps to, from=2-3, to=2-4]
  \arrow["\cost{1}", maps to, from=2-4, to=2-5]
  \arrow["\cost{1}", maps to, from=2-5, to=2-6]
  \arrow["\cost{1}", maps to, from=2-6, to=2-7]
  \arrow["\cost{1}", maps to, from=2-7, to=2-8]
  \arrow["\cost{1}", maps to, from=2-8, to=2-9]
  \arrow["\cost{0}", maps to, color={rgb,255:red,127;green,127;blue,127}, from=1-1, to=2-1]
  \arrow["\cost{1}", maps to, color={rgb,255:red,127;green,127;blue,127}, from=1-2, to=2-2]
  \arrow["\cost{2}", maps to, color={rgb,255:red,127;green,127;blue,127}, from=1-3, to=2-3]
  \arrow["\cost{3}", maps to, color={rgb,255:red,127;green,127;blue,127}, from=1-4, to=2-4]
  \arrow["\cost{4}", maps to, color={rgb,255:red,127;green,127;blue,127}, from=1-5, to=2-5]
  \arrow["\cost{5}", maps to, color={rgb,255:red,127;green,127;blue,127}, from=1-6, to=2-6]
  \arrow["\cost{6}", maps to, color={rgb,255:red,127;green,127;blue,127}, from=1-7, to=2-7]
  \arrow["\cost{7}", maps to, color={rgb,255:red,127;green,127;blue,127}, from=1-8, to=2-8]
  \arrow["\cost{0}", maps to, color={rgb,255:red,127;green,127;blue,127}, from=1-9, to=2-9]
\end{tikzcd}\]
Intuitively, the vertical arrows keep track of the difference between the realistic and imagined cost.
At the beginning and end of this trace, both perspectives align and the difference is $\cost{0}$, but at intermediate states, the imagined bottom perspective has incurred more cost than the realistic top perspective.
This difference, called the \emph{potential}, is essential in amortized analysis.

The technique of amortized analysis has been widely applied to data structures since, giving more practical bounds on ephemeral data structures.
Since its inception, the amortized study of data structures has been phrased algebraically, reasoning about the total cost of a finite sequence of operations:
\begin{quote}
  In many uses of data structures, a \emph{sequence of operations}, rather than just a single operation, is performed, and we are \emph{interested in the total time of the sequence}, rather than in the times of the individual operations.~\cite{tarjan>1985}
\end{quote}
This algebraic emphasis is in contrast to the common coalgebraic semantics taken when giving a semantics to sequential-use data structures, sometimes referred to as ``objects''~\cite{jacobs>1996-objects-coalgebraically}.
In this work, we take the perspective that amortized analysis is fundamentally coalgebraic, showing that the techniques used in amortized analysis are specialized instances of more general coalgebraic machinery, which elegantly connects to the theory of synthetic cost and behavior verification.

\subsection{Amortized Analysis}\label{sec:telescoping}

In the original development of amortized analysis, Tarjan and Sleator~\cite[\S 2]{tarjan>1985} describe the \emph{physicist's method}, a general technique for tracking amortized cost in a sequence of operations.
Let $\implObj$ be the set of states of a data structure.
In this method, one defines a function $\Phi : \implObj \to \Int$ that assigns to each state a \emph{potential}, representing the difference (assumed to be nonnegative) between the realistic and imagined cost models.
Then, if $\behof{\implCoalg} : \implObj \to \implObj$ implements an operation, the \emph{amortized cost} of $\behof{\implCoalg}$ on a state $d$ is defined as
\[ \costof{\specCoalg} = \costof{\implCoalg}(d) + \Phi(\behof{\implCoalg}(d)) - \Phi(d), \]
where $\costof{\implCoalg}(d)$ is the true cost of the operation on state $d$.
It is common to iterate this amortization equation to reason about the total cost of a finite-length sequence of operations.
Letting $d^{(i)} = \behof{\implCoalg}^{(i)}(d)$ and $\costof{\specCoalg}^{(i)} = \costof{\implCoalg}(d^{(i)}) + \Phi(d^{(i+1)}) - \Phi(d^{(i)})$, we have the following by telescoping sums:
\begin{align*}
  \sum_{i = 0}^{n - 1} \costof{\specCoalg}^{(i)}
    = (\Phi(d^{(n)}) - \Phi(d^{(0)})) + \sum_{i = 0}^{n - 1} \costof{\implCoalg}(d^{(i)})
\end{align*}
Then, if $\Phi(d^{(0)}) \le \Phi(d^{(n)})$, the true total cost $\sum_i \costof{\implCoalg}(i)$ is bounded by the amortized cost $\sum_i \costof{\specCoalg}^{(i)}$.
Given a suitable choice of $\Phi$, each amortized cost $\costof{\specCoalg}^{(i)}$ can often be bounded by a simple term, such as a constant.
For example, if $\costof{\specCoalg}^{(i)}$ is always a constant $k$, then the amortized cost of the sequence of operations is $kn$.

\subsection{Coalgebraic Semantics of Data Structures}\label{sec:intro-coalgebra}

Coalgebras have been used to give a semantics for data structures implementing sequential-use abstract data types, in the style of object-oriented programming~\cite{jacobs>1996-objects-coalgebraically}.
A \emph{signature} (or interface) is represented by an endofunctor $\sigFun : \catC \to \catC$, where $\sigFun \X$ represents the operations provided by an implementation when $\X$ is the implementation type.
For example, when $\catC = \Set$, the signature
\[ \sigFun \X = (E \Rightarrow \X) \times (1 + (E \times \X)) \]
describes data structures $\X$ that export two methods: one of type $E \Rightarrow \X$ and one of type $1 + (E \times \X)$.
For example, this signature could be used to represent stacks or queues, where the methods are either push and pop (for stacks) or enqueue and dequeue (for queues).
A \emph{$\sigFun$-coalgebra} is a pair $(\implObj, \implCoalg)$ of a \emph{carrier} object $\implObj : \catC$ and a \emph{transition morphism} $\implCoalg : \implObj \to \sigFun\implObj$.
Such a coalgebra should be understood as an implementation of the signature $\sigFun$, consisting of a state type $\implObj$ and an implementation of the methods via $\implCoalg$.
When $\sigFun$ is a product, as above, $\implCoalg$ can be specified via a collection of maps to each component, using the universal property of products.
For example, here a coalgebra consists of two maps:
\begin{mathpar}
  {\implCoalg_1 : \implObj \to E \Rightarrow \implObj}

  {\implCoalg_2 : \implObj \to 1 + (E \times \implObj)}
\end{mathpar}
Given a state of type $\implObj$, the maps offer each method for use.
For example, interpreting the above signature for stacks, the methods will implement push and pop, respectively.

\subsection{Abstract Cost Analysis via the Writer Monad}

This work builds on the proposal of Grodin~and~Harper~\cite{grodin-harper>2023} to study amortized analysis coalgebraically in Calf, an effectful dependent type theory based on call-by-push-value~\cite{levy>2003} that supports the verification of both correctness conditions and cost bounds~\cite{niu-sterling-grodin-harper>2022,grodin-niu-sterling-harper>2024}.
We recall the following syntax for programming with $\F{}$-types, writing $\A, \B, \C$ for value types and $\X, \Y, \Z$ for computation types:
\begin{mathpar}
  \infer
    {V : \A}
    {\ret{V} : \F{\A}}

  \infer
    {M : \F{\A} \\ x : \A \vdash M' : \X}
    {\bindex{M}{x}{M'} : \X}
\end{mathpar}
We also use coproducts and monoidal products of computation types, whose syntaxes are presented in the Enriched Effect Calculus~\cite{egger-mogelberg-simpson>2009,egger-mogelberg-simpson>2014} and Linear/Non-Linear type theory~\cite{benton>1995}, respectively.
As in Calf, we include an effect primitive for instrumenting a computation $M$ with $c$ units of abstract cost, here notated $\charge{c}{M}$, where $c : \Cost$ and $(\Cost, +, 0)$ is a monoid representing cost.
For clarity, we notate that a value $c : \Cost$ is a cost as $\cost{c}$.
The cost effect increments an ongoing counter by $c$, respecting the monoid structure:
\begin{mathpar}
  \charge{0}{M} = M

  \charge{c_1}{\charge{c_2}{M}} = \charge*{c_1 + c_2}{M}
\end{mathpar}
Crucially for amortization, effects in call-by-push-value commute with other computations, which can be understood by thinking of computation types as ``lazy'' types of call-by-name:
\begin{align*}
  \bindex{\charge{c}{M}}{x}{M'} &= \charge{c}{\bindex{M}{x}{M'}} & (M : \F{\A}) \\
  \charge{c}{M}~\proj{proj}_i &= \charge{c}{M~\proj{proj}_i} & (M : \X \times \Y) \\
  \charge{c}{M}~V &= \charge{c}{M~V} & (M : \A \pto \X) \\
  \charge{c}{V, M} &= (V, \charge{c}{M}) & ((V, M) : \A \rtimes \X)
\end{align*}
One key observation of Calf is that cost analysis is inseparable from correctness verification: in general, the cost of a program may depend on its behavior.
This attitude is further validated in the present work, since the amortized cost of a data structure operation may depend on aspects of its state.

Semantically, we will work in the Eilenberg--Moore category of a monad $T$ on $\Set$, written $\Alg{T}$.
This category is complete and cocomplete, and it is powered and copowered over $\Set$.
We recall that for any monad $T : \Set \to \Set$, there is an adjunction $\F{} \dashv \U{} : \Alg{T} \to \Set$ whose induced monad is $T$.
Often, we will let $T$ be a writer monad $\Writer{(-)}$.
An algebra for the writer monad is a set $\U{\X}$ equipped with a coherent method for storing abstract cost within the set, $\mathsf{charge}_{\X} : \Cost \times \U{\X} \to \U{\X}$.
In the present work, this aspect of the category of algebras will be essential: since every object comes equipped with a method for absorbing cost, any cost incurred will by construction be amortized forward.
Inspired by call-by-push-value~\cite{levy>2003}, we will abbreviate morphisms $\F{\A} \to \X$ as simply $\A \pto \X$, which we often implicitly understand as $\A \to \U{\X}$ using the adjunction.
Such maps implicitly propagate cost from the input to the output, also essential for amortization since previously-incurred cost should never be lost.
Also, when using the writer monad, we will notate the cost and behavior components of a map $\implCoalg : \A \pto \F{\B}$ as $\costof{\implCoalg} : \A \to \Cost$ and $\behof{\implCoalg} : \A \to \B$, respectively.

\paragraph{Synopsis}
In Section~\ref{sec:coalgebra}, we introduce coalgebra morphisms and show that in the presence of cost, they generalize the potential functions of amortized analysis.
In Section~\ref{sec:lax}, we support inexact amortized upper bounds by generalizing from categories to bicategories and from coalgebra morphisms to colax coalgebra morphisms.
In Section~\ref{sec:profunctor}, we incorporate mixed-variance operations using endoprofunctors, exploiting the symmetric monoidal structure on the category of writer monad algebras present given a commutative cost model to sum the potential of all relevant instances.
In Section~\ref{sec:composition}, we observe that potential functions can be composed, and we view coalgebras as an indexed category in order to implement one amortized data structure in terms of another with a differing signature.
\section{Coalgebra Morphisms as Generalized Potential Functions}\label{sec:coalgebra}

For the purpose of cost analysis, we will typically consider data structures implemented as coalgebras over a signature endofunctor $\sigFun$ on $\Alg{\Writer{(-)}}$, the category of writer monad algebras for a monoid $(\Cost, +, 0)$.
The carrier of the coalgebra in our examples will typically be of the form $\F{\underline{\implObj}}$, the free cost algebra on a set $\underline{\implObj}$.
We start by giving two examples of data structure implementations, each with only a single method.

\begin{example}\label{ex:allocation}
  For this example, we use $(\Cost, +, 0) \isdef (\Int, +, 0)$, the additive monoid of integers.
  We use the signature functor $\sigFun = \Id : \Alg{\Writer{(-)}} \to \Alg{\Writer{(-)}}$ to represent a single method: a transition morphism $\implCoalg : \implObj \to \Id(\implObj)$ is the implementation of the method for chosen state type $\implObj$.
  An $\Id$-coalgebra consists of an object $\implObj$ of $\Alg{\Writer{(-)}}$ and a transition morphism $\implCoalg : \implObj \to \implObj$.
  Treating the method as allocation of one unit of space, we may define two $\Id$-coalgebras: one simple, unrealistic model that allocates one cell per call, and one realistic model that allocates eight spaces every eight calls.
  \begin{enumerate}
    \item
      The simple, unrealistic \emph{specification} coalgebra has carrier $\specObj = \F{1}$ and transition morphism as follows:
      \begin{center}
        \iblock{
          \mrow{\specCoalg : 1 \pto \F{1}}
          \mrow{\specCoalg~\triv = \charge{1}{\ret{\triv}}}
        }
      \end{center}
      In this coalgebra, no state is maintained, and one allocation is performed per transition.
    \item
      The realistic \emph{data structure} coalgebra has carrier $\implObj = \F*{\Fin{8}}$, tracking how many already-allocated cells are free.
      Its transition morphism is given by:
      \begin{center}
        \iblock{
          \mrow{\implCoalg : \Fin{8} \pto \F*{\Fin{8}}}
          \mrow{\implCoalg~\zero = \charge{8}{\ret{7}}}
          \mrow{\implCoalg~(\suc{d}) = \ret{d}}
        }
      \end{center}
      If no space remains, $\cost{8}$ of cost is charged, allocating eight cells; seven cells remain after the transition.
      If some space remains, the amount of space remaining is decreased without performing any allocations.
  \end{enumerate}
  The coalgebra $(\implObj, \implCoalg)$ is an amortizing implementation of the specification coalgebra $(\specObj, \specCoalg)$.
  The coalgebras are not always exactly synchronized, as (starting from state $7 : \Fin{8}$) the implementation incurs a large cost only after the specification ``saves up'' enough cost to match it.
  However, from the perspective of a client, the more complex implementation can be approximated via the simple specification.
\end{example}

To prove the relationship between the specification $(\specObj, \specCoalg)$ and the amortized implementation $(\implObj, \implCoalg)$, we write a morphism of coalgebras, generalizing the potential functions of amortized analysis.

\subsection{Coalgebra Morphisms}

To prove that $(\implObj, \implCoalg)$ is an amortizing implementation of $(\specObj, \specCoalg)$, one classically gives a potential function $\Phi : \implObj \to \Cost$ satisfying the amortization condition discussed in Section~\ref{sec:introduction}.
We now define what it means to be a morphism of coalgebras and show that the amortization condition falls out as a special case when working in the category of writer monad algebras.

\begin{definition}
  Let $(\implObj, \implCoalg)$ and $(\specObj, \specCoalg)$ be $\sigFun$-coalgebras.
  A \emph{morphism of $\sigFun$-coalgebras} from $(\implObj, \implCoalg)$ to $(\specObj, \specCoalg)$ consists of a morphism $\Phi : \implObj \to \specObj$ that preserves the $\sigFun$-coalgebra structure:
  \[\begin{tikzcd}
    \implObj & \sigFun \implObj \\
    \specObj & \sigFun \specObj
    \arrow["\implCoalg", from=1-1, to=1-2]
    \arrow["\specCoalg", from=2-1, to=2-2]
    \arrow["\Phi", from=1-1, to=2-1]
    \arrow["\sigFun\Phi", from=1-2, to=2-2]
  \end{tikzcd}\]
  In other words, $\Phi$ preserves observational equivalence: using only the $\sigFun$-coalgebra structure, one can make identical observations regardless of when $\Phi$ is used to transition from $\implObj$ to $\specObj$.
  In this sense, $(\implObj, \implCoalg)$ is simulated by $(\specObj, \specCoalg)$, with the simulation mediated by $\Phi$: up to the translation $\Phi$, the coalgebra $(\specObj, \specCoalg)$ behaves just like $(\implObj, \implCoalg)$.
  We will refer to the commutativity of this diagram as the \emph{generalized amortization condition} for reasons that will be developed shortly.
\end{definition}

We write $\Coalg{\sigFun}$ for the category of $\sigFun$-coalgebras and coalgebra morphisms.
Now, we show that the requirement that $\Phi$ preserve the coalgebra structure is exactly the required condition on a potential function in amortized analysis.

\begin{example}\label{ex:allocation-potential}
  Recall the $\Id$-coalgebras from Example~\ref{ex:allocation}.
  To give a morphism from $(\implObj, \implCoalg)$ to $(\specObj, \specCoalg)$, we must provide a function $\Phi : \F*{\Fin{8}} \to \F{1}$, equivalently written $\Phi : \Fin{8} \pto \F{1}$, that preserves the coalgebra structure, as specified above.
  Equationally, the structure preservation condition says that \[ \Phi \semi \specCoalg = \implCoalg \semi{} \sigFun\Phi \] as morphisms $\Fin{8} \pto \F{1}$.
  Using the Eilenberg--Moore adjunction, a map $\Fin{8} \pto \F{1}$ is equivalent to function $\Fin{8} \to \U*{\F{1}}$. Since $\U*{\F{1}} = \Writer{1} \cong \Cost$, it is equivalent to give a function $\Phi : \Fin{8} \to \Cost$ such that for all $d : \Fin{8}$, it is the case that \[ \Phi(d) + \costof{\specCoalg}(\triv) = \costof{\implCoalg}(d) + \Phi(\behof{\implCoalg}(d)). \]
  Subtracting $\Phi(d)$ from both sides, using the commutativity of addition in $\Int$, and dropping the trivial argument $\triv : 1$, this condition is exactly the amortization condition
  \[ \costof{\specCoalg} = \costof{\implCoalg}(d) + \Phi(\behof{\implCoalg}(d)) - \Phi(d) \]
  discussed in Section~\ref{sec:introduction}.
  In this case, such a function can be defined as $\Phi(d) = \cost*{7 - d}$.
  This morphism of coalgebras precisely witnesses the fact that $(\implObj, \implCoalg)$ is an amortizing implementation of $(\specObj, \specCoalg)$.
\end{example}

Here, $(\specObj, \specCoalg)$ serves as a client-facing amortized cost specification for $(\implObj, \implCoalg)$.
While no meaningful information about the state of the computation is provided, the specification conveys an amortized cost that is accurate up to the potential function $\Phi$.
By including more information in the carrier of the specification coalgebra, we may support amortized costs that vary over time.

\begin{example}
  Often, the cost of operations varies over time, whereas the previous example considers only constant specification costs.
  In simple traditional amortized analyses, one uses functions $\costof{\implCoalg} : \Nat \to \Cost$ and $\costof{\specCoalg} : \Nat \to \Cost$ to assign distinct costs to each operation in a sequence, asking that
  \[ \costof{\specCoalg}(i) = \costof{\implCoalg}(i) + \Phi(i + 1) - \Phi(i) \]
  where $\Phi : \Nat \to \Cost$.
  Letting $\specObj = \implObj = \F{\Nat}$ and $\behof{\specCoalg}(i) = \behof{\implCoalg}(i) = i + 1$, we recover the equivalent coalgebra morphism condition,
  \[ \Phi(i) + \costof{\specCoalg}(i) = \costof{\implCoalg}(i) + \Phi(i + 1). \]
  Here, both transition morphisms maintain the index of the current operation, making the assumption that some desired sequence of operations has been specified \textit{a priori}.
  More generally, though, our coalgebraic perspective allows for arbitrary carriers, viewing the cost model of a data structure alongside its implementation rather than attempting to extract a cost-only function as a separable quantity.
\end{example}

Traditional accounts of amortized analysis make implicit use of commutativity and additive inverses in $\Cost$.
However, the generalized amortization condition is sensible regardless of the cost model.
Reasoning principles pertaining to amortization can be expressed at this new level of generality.
For example, recall from Section~\ref{sec:telescoping} that a telescoping sum can be used to bound the cost of an $n$-length sequence of operations.
Coalgebraically, this condition is simply the composition of $n$ coalgebra morphism squares:
\[\begin{tikzcd}
  \implObj & \sigFun\implObj & {\sigFun^{(2)}\implObj} & \cdots & {\sigFun^{(n)}\implObj} \\
  \specObj & \sigFun\specObj & {\sigFun^{(2)}\specObj} & \cdots & {\sigFun^{(n)}\specObj}
  \arrow["\Phi", from=1-1, to=2-1]
  \arrow["\implCoalg", from=1-1, to=1-2]
  \arrow["\specCoalg", from=2-1, to=2-2]
  \arrow["\sigFun\Phi", from=1-2, to=2-2]
  \arrow["{\sigFun^{(n-1)}\specCoalg}", from=2-4, to=2-5]
  \arrow["{\sigFun^{(n-1)}\implCoalg}", from=1-4, to=1-5]
  \arrow["{\sigFun^{(n)}\Phi}", from=1-5, to=2-5]
  \arrow[from=1-3, to=1-4]
  \arrow[from=2-3, to=2-4]
  \arrow["\sigFun\specCoalg", from=2-2, to=2-3]
  \arrow["\sigFun\implCoalg", from=1-2, to=1-3]
  \arrow["{\sigFun^{(2)} \Phi}", from=1-3, to=2-3]
\end{tikzcd}\]
Here, $\sigFun^{(n)}$ is the $n$-fold composition of $\sigFun$.
Observe that when using the coalgebras of \ref{ex:allocation} and the potential function of \ref{ex:allocation-potential}, the commutative mappings presented in Section~\ref{sec:introduction} are induced by this diagram.
Henceforth, we break away from $\Int$ in favor of a cost model like $\Nat$ that does \emph{not} admit additive inverses, making precise the assumption that potential must be nonnegative.

\subsection{Classic Amortized Analyses}

We will now describe more complex analyses using this framework.
In general, for a data structure implementation $(\implObj, \implCoalg)$ and an amortized cost specification $(\specObj, \specCoalg)$, an amortized analysis will be a coalgebra morphism
$\Phi : (\implObj, \implCoalg) \to (\specObj, \specCoalg)$
serving as a behavior-relevant generalization of potential functions.

\begin{example}\label{ex:dynamic-array}
  Let $\sigFun \X = E \pto \X$, the $E$-fold power, representing the signature with one method for reading one value of type $E$ at a time.
  We can use this signature to implement a dynamically-resizing array, a classic first example of an amortized data structure, where the method represents pushing an element of type $E$ to the end of the array.
  In this data structure, we represent a list of arbitrary length using an underlying array with a fixed length.
  When more data is added to the list than the current array can hold, a new array of twice the length is allocated, and the old data is copied over.
  In the cost model here, we charge one unit of cost for each write, array allocation, and array deallocation.
  \begin{enumerate}
    \item
      The specification again uses carrier $\specObj = \F{1}$.
      Its transition morphism $\specCoalg : 1 \pto \F{1}$ is given by:
      \begin{center}
        \iblock{
          \mrow{\specCoalg : \F{1}}
          \mrow{\specCoalg = \charge{3}{\ret{\triv}}}
        }
      \end{center}
      In this coalgebra, no state is maintained, and we charge $\cost{3}$ at each operation in order to save up for an eventual copy.

    \item
      The data structure implementation uses carrier $\implObj = \F{\underline{\implObj}}$, where
      \[ \underline{\implObj} = \sum_{n : \Nat} \mathsf{array}_E[2^n - 1, 2^{n+1}-1) \]
      stores a logarithm-size bound $n$ and an array with length between $2^n - 1$ and $2^{n + 1} - 1$.
      We give the transition morphism $\implCoalg$ by cases:
      \begin{center}
        \iblock{
          \mrow{\implCoalg : \underline{\implObj} \pto E \pto \F{\underline{\implObj}}}
          \mrow{\implCoalg~(n, a)~e =
            \begin{cases}
              \charge*{3 + \length{a}}{\ret{\suc{n}, \consright{a}{e}}} & \text{if } \length{a} + 1 = 2^{n + 1} - 1 \\
              \charge{1}{\ret{n, \consright{a}{e}}} & \text{otherwise}
            \end{cases}}
        }
      \end{center}
      We charge $\cost*{3 + \length{l}}$ in the expensive case when we have to copy the data to a new array, accounting for the allocation, write, copy, and deallocation, and we charge $\cost{1}$ for each write in the cheap case, since we have enough space and only need to perform a single write operation.
  \end{enumerate}
  To give a morphism from $(\implObj, \implCoalg)$ to $(\specObj, \specCoalg)$, we must provide a function $\Phi : \underline{\implObj} \to \Nat$ such that the necessary square commutes.
  We define $\Phi(n, a) = \cost*{2(\length{a} + 1) - 2^{n + 1}}$, which meets the necessary criterion, showing that $(\specObj, \specCoalg)$ is a reasonable specification for $(\implObj, \implCoalg)$.
\end{example}

Once again, the carrier of the specification coalgebra is trivial, since each push operation has a constant amortized cost.
Sometimes, though, an operation can have a cost dependent on some aspects of the state, modeled by a nontrivial carrier of the specification coalgebra.
Then, the generalized potential function reveals these aspects about the data structure state in addition to mediating the potential cost.

\begin{example}
  Extending Example~\ref{ex:dynamic-array}, we can add a method parameterized by an endofunction $E \Rightarrow E$ to update all values in an array.
  Let $\sigFun \X = (E \Rightarrow \X) \times ((E \Rightarrow E) \Rightarrow \X)$.
  Then, for a carrier $\implObj$, a transition morphism consists of a pair of maps $\implCoalg_1 : \implObj \to E \Rightarrow \implObj$ and $\implCoalg_2 : \implObj \to ((E \Rightarrow E) \Rightarrow \implObj)$.
  We let $\implCoalg_1$ be the same map as before, and we let $\implCoalg_2(n,a)(f) = \charge*{\length{a}}{\ret{n, \mathsf{map}~f~a}}$.
  There is no way to amortize this large cost, since the update method may be called arbitrarily often.
  Thus, to match the cost of this update method in the specification coalgebra $(\specObj, \specCoalg)$, the carrier must reveal some data about the underlying array.
  Since only the length of the array matters here, we may choose $\specObj = \F{\Nat}$, only tracking the length.
  Then, we define $\specCoalg$ as follows, giving the projections separately via copattern matching~\cite{abel-pientka-thibodeau-setzer>2013}:
  \begin{center}
    \iblock{
      \mrow{\specCoalg : \Nat \pto \sigFun(\F{\Nat})}
      \mrow{\specCoalg~\proj{push}~n~e = \charge{3}{\ret{\suc{n}}}}
      \mrow{\specCoalg~\proj{update}~n~f = \charge{n}{\ret{n}}}
    }
  \end{center}
  The push method still costs $\cost{3}$, but now it must remember the increase of the length of the array.
  The update method now can cost $\cost{n}$ when the state is some $n$, preserving the length $n$.
  The coalgebra morphism $\Phi : \underline{\implObj} \pto \F{\Nat}$ is identical on cost, but the new nontrivial behavior component sends a bounded array $(n, a)$ to its length, $\length{a}$, written $\Phi(n, a) = \charge*{2(\length{a} + 1) - 2^{n + 1}}{\ret{\length{a}}}$.
\end{example}

In this example, the state-dependent cost only appears for the non-amortized operation.
In general, though, this need not be the case; cost is allowed to depend on behavior.
For example, the operations on a splay tree have an amortized cost logarithmic in the size of the tree~\cite{sleator-tarjan>1985}.
We now consider queues, a classic example of an amortized data structures that amortize cost using two operations.

\begin{example}\label{ex:queue}
  A queue is an abstract data type in which elements are enqueued to one end of the queue and dequeued from the other end.%
  \footnote{This was the principal example of the predecessor to this work~\cite{grodin-harper>2023}.}
  To express this method in the signature, we let \[ \sigFun \X = (E \pto \X) \times (\F{1} + E \rtimes \X), \] where $E \rtimes \X$ is the $E$-fold copower.
  The first method, enqueue, accepts an element of type $E$ to store and continues.
  The second method, dequeue, either terminates if the queue is empty or provides an element of type $E$ before continuing.
  This signature is similar to that of Section~\ref{sec:intro-coalgebra}, adapted to the category of cost algebras by replacing exponentials and products with powers and copowers where necessary.

  \begin{enumerate}
    \item
      The specification uses carrier $\specObj = \F*{\listty{E}}$, storing a specification-level list of elements.
      Its transition morphism $\specCoalg$ implements queues naively via a list of elements:
      \begin{center}
        \iblock{
          \mrow{\specCoalg : \listty{E} \pto \sigFun(\F*{\listty{E}})}
          \mrow{\specCoalg~\proj{enqueue}~l~e = \charge{1}{\ret{\consright{l}{e}}}}
          \mrow{\specCoalg~\proj{dequeue}~\nilex = \inj{1}(\ret{\triv})}
          \mrow{\specCoalg~\proj{dequeue}~(\consex{e}{l}) = \inj{2}(e, \ret{l})}
        }
      \end{center}

    \item
      The data structure implementation uses carrier $\implObj = \F*{\listty{E}^2}$, storing a pair of lists to form a ``batched queue'' \cite{hood-melville>1981,burton>1982,gries>1989,okasaki>1999}.
      The first list is treated as an ``inbox'', storing enqueued elements, and the second list is treated as an ``outbox'', producing elements to dequeue.
      Occasionally, when the outbox is empty, the inbox is reversed and placed in the outbox.
      \begin{center}
        \iblock{
          \mrow{\implCoalg : \listty{E}^2 \pto \sigFun(\F*{\listty{E}^2})}
          \mrow{\implCoalg~\proj{enqueue}~(l_i, l_o)~e = \ret{\consex{e}{l_i}, l_o}}
          \mrow{\implCoalg~\proj{dequeue}~(l_i, \nilex) =
            \begin{cases}
              \inj{1}(\ret{\triv}) & \text{if }\mathsf{reverse}(l_i) = \nilex \\
              \charge*{\length{l_i}}{\inj{2}(e, \ret{l_o})} & \text{if }\mathsf{reverse}(l_i) = \consex{e}{l_o}
            \end{cases}
          }
          \mrow{\implCoalg~\proj{dequeue}~(l_i, \consex{e}{l_o}) = \inj{2}(e, \ret{l_i, l_o})}
        }
      \end{center}
      Here, our cost model charges $\cost*{\length{l}}$ cost for a call to $\mathsf{reverse}(l)$.
  \end{enumerate}
  The coalgebra morphism representing the amortized analysis is a first-class effectful program, integrating cost and behavior verification:
  \[ \Phi(l_i, l_o) = \charge*{\length{l_i}}{\ret{l_o \mdoubleplus \mathsf{reverse}(l_i)}} \]
  In addition to the traditional cost-level potential function $\length{l_i}$, it includes a behavioral simulation $l_o \mdoubleplus \mathsf{reverse}(l_i)$ converting the pair of lists to a single specification-level list.
\end{example}

\subsection{Generalizations of Amortized Analysis}

Viewing amortized analysis coalgebraically allows the underlying category to be swapped out, leading to compact and elegant presentations of novel variations of amortization.

\begin{example}
  Traditionally, it is assumed that addition in the cost model is commutative and admits an inverse.
  In this form, though, no requirements are placed on the monoid whatsoever.
  For example, we may let our ``costs'' be strings, where the monoid operation is concatenation.
  Then, amortization represents buffering, a performance technique in which many strings are occasionally printed in aggregate to avoid repeating fixed costs associated with the writing of any data.
  Let $\sigFun \X = \String \pto \X$, providing a single method for printing.
  \begin{enumerate}
    \item
      The specification coalgebra uses state $\specObj = \F{1}$.
      Its transition morphism $\specCoalg : 1 \pto \String \pto \F{1}$ simply prints the provided string:
      \begin{center}
        \iblock{
          \mrow{\specCoalg : \String \pto \F{1}}
          \mrow{\specCoalg~s = \charge{s}{\ret{\triv}}}
        }
      \end{center}
      Here, ``charging a string cost'' should be understood as printing the string.
    \item
      The implementation uses state $\implObj = \F{\underline{\implObj}}$, where \[ \underline{\implObj} = \sum_{s : \String} \length{s} < n \]
      for a fixed buffer size $n$.
      Its transition morphism prints strings in chunks of length $n$, saving any remaining characters in the state.
      Let $\mathsf{chop}_n$ split a string into a portion with length a multiple of $n$ and a remainder with length less than $n$.
      Then, we define:
      \begin{center}
        \iblock{
          \mrow{\implCoalg : \underline{\implObj} \pto \String \pto \F{\underline{\implObj}}}
          \mrow{\implCoalg~s_0~s = \mathsf{let}~(s', s'_0) = \mathsf{chop}_n(s_0 \mdoubleplus s)~\mathsf{in}~\charge{s'}{\ret{s'_0}}}
        }
      \end{center}
  \end{enumerate}
  The generalized potential function $\Phi : \underline{\implObj} \to \String$ is simply the inclusion, ``flushing'' any data remaining in the buffer.
  For example, when $n = 8$ and printing $\mathsf{``world"}$ on a buffer containing $\mathsf{``hello"}$,  the generalized amortization condition is the following:
  \[ \Phi(\mathsf{``hello"}) \mdoubleplus \mathsf{``world"} = \mathsf{``hellowor"} \mdoubleplus \Phi(\mathsf{``ld"}) \]
  In other words, we ask that flushing the original buffer with $\mathsf{``hello"}$ and then printing $\mathsf{``world"}$ is equivalent to printing $\mathsf{``hellowor"}$ via the buffered implementation and then flushing the remaining buffer $\mathsf{``ld"}$.

  This shows that buffering can be understood as an amortized implementation of printing strings in real time.
  Using the state monad in place of the writer monad, this technique can be adapted to support buffering of arbitrary state.
\end{example}

From this perspective, we may move beyond the writer monad, amortizing other effects.
For any strong monad $T$, it is also the case that $T(\Writer{(-)})$ forms a monad, where $\Cost$ is an arbitrary monoid for cost.
We now consider working with coalgebras in the category of $T(\Writer{(-)})$-algebras for various choices of $T$.

\begin{example}
  When $T = \mathcal{D}$ is the finitely-supported distribution monad, we can study randomized amortized analysis.
  For example, letting $\sigFun = \Id$, we can implement an alternating amortization technique that flips many coins occasionally, whereas the specification suggests that one coin is flipped per transition.
  \begin{enumerate}
    \item
      The specification coalgebra uses state $\specObj = \F{1}$.
      Its transition morphism $\specCoalg : 1 \pto \F{1}$ is a simple Bernoulli distribution, flipping a coin and deciding whether to incur cost accordingly.
    \item
      The implementation coalgebra uses state $\implObj = \F*{\Fin{k}}$ for a fixed $k$ indicating how often to sample.
      Its transition morphism $\implCoalg : \Fin{k} \pto \F*{\Fin{k}}$ is like that of Example~\ref{ex:allocation}, sampling a $k$-binomial distribution when its counter is $0$ and decrementing the counter otherwise.
  \end{enumerate}
  The potential function $\Phi : \Fin{k} \pto \F{1}$ computes a distribution for each state $d : \Fin{k}$.
  Similar to Example~\ref{ex:allocation-potential}, we let the generalized potential function be the $(k - d - 1)$-binomial distribution, balancing the number of samples done by the specification and the implementation.
  Notice that here, though, $\Phi$ is not merely computing a number as a potential: it computes an entire distribution via an effectful program.
\end{example}

\begin{example}
  Sometimes, one may consider expected amortized analysis of a randomized data structure, reasoning about the expected cost of a sequence of operations.
  While this notion is subtle to define explicitly,
  a reasonable definition of expected amortized analysis falls out of the coalgebraic perspective.
  If the cost model $\Cost$ is equipped with the structure of a convex space (\ie, a $\mathcal{D}$-algebra structure),
  then we have a distributive law that computes the expected value:
  \[ \mathcal{D}(\Writer{(-)}) \to \Writer{\mathcal{D}(-)} \]
  Therefore, $\Writer{\mathcal{D}(-)}$ forms a monad for reasoning about expected cost, and coalgebras over an endofunctor on $\Alg{\Writer{\mathcal{D}(-)}}$ cleanly and precisely specify expected amortized analysis.
\end{example}

Although the remainder of the paper is compatible with other monads, we work only with the writer monad with numeric costs for simplicity of examples.
\section{Lax Amortized Analysis in a Bicategory}\label{sec:lax}

In some examples, such as those considered thus far, amortized analysis is precise, where the specification exactly matches the amount of cost used by the implementation.
However, it is common for the specified cost to be an upper bound, since some cases may be cheaper than specified due to internal factors that would be difficult to communicate via a concise specification.
In the literature, the equation
\[ \costof{\specCoalg} = \costof{\implCoalg}(d) + \Phi(\behof{\implCoalg}(d)) - \Phi(d) \]
defines the amortized cost $\costof{\specCoalg}$, which is then given an upper bound.
Since here $(\specObj, \specCoalg)$ is a specification implementation, though, we treat $\specCoalg$ as the upper bound itself, turning the above equation into the inequality
\[ \costof{\specCoalg} \ge \costof{\implCoalg}(d) + \Phi(\behof{\implCoalg}(d)) - \Phi(d). \]
To achieve this categorically, we augment our development slightly: rather than using coalgebras and morphisms for an endofunctor on a category, we consider \emph{colax} coalgebras and morphisms for an endo\emph{-2-}functor on a \emph{bicategory}.
The 2-cells of the bicategory will serve as inequalities for the purpose of cost analysis, drawing inspiration from Grodin~\etal~\cite{grodin-niu-sterling-harper>2024}.

For a bicategory $\catC$, let $\sigFun : \catC \to \catC$ be an endo-2-functor.
Although the definition of $\sigFun$-coalgebra does not change, the definition of morphism between $\sigFun$-coalgebras is affected: rather than considering coalgebra morphisms where the given square commutes exactly, we only ask for the square to commute up to a 2-cell.

\begin{definition}
  Let $(\implObj, \implCoalg)$ and $(\specObj, \specCoalg)$ be $\sigFun$-coalgebras.
  A \emph{colax morphism of $\sigFun$-coalgebras} \cite{blackwell-kelly-power>1989,lack>2005,lack-shulman>2012} from $(\implObj, \implCoalg)$ to $(\specObj, \specCoalg)$ consists of a morphism $\Phi : \implObj \to \specObj$ that colaxly preserves the $\sigFun$-coalgebra structure:
  \[\begin{tikzcd}
    \implObj & \sigFun \implObj \\
    \specObj & \sigFun \specObj
    \arrow["\implCoalg", from=1-1, to=1-2]
    \arrow["\specCoalg", from=2-1, to=2-2]
    \arrow["\Phi", from=1-1, to=2-1]
    \arrow["\sigFun\Phi", from=1-2, to=2-2]
		\arrow["\varphi"{description}, shorten <=4pt, shorten >=4pt, Rightarrow, from=1-2, to=2-1]
  \end{tikzcd}\]
	Here, $\varphi$ is a 2-cell $\Phi \semi \specCoalg \Leftarrow \implCoalg \semi{} \sigFun\Phi$, serving as a proof of inequality.
	We refer to $\varphi$ as the \emph{lax generalized amortization condition}.
\end{definition}

For amortized analysis, we will typically find ourselves using a 2-poset, a bicategory whose 2-cells are mere propositions.
In this case, the 2-cell of a morphism is simply a proof that $\Phi \semi \specCoalg \ge \implCoalg \semi{} \sigFun\Phi$.
Concretely, we will choose $\catC = \Poset$.
Following Grodin~\etal~\cite{grodin-niu-sterling-harper>2024}, we will use discrete posets everywhere except for the cost component on the writer monad, which will be the natural numbers equipped with the usual increasing ordering, written $\omega$.

\begin{example}\label{ex:colax}
	Let $(\F{\underline{\implObj}}, \implCoalg)$ and $(\F{1}, \specCoalg)$ be $\Id$-coalgebras.
	Then, a colax morphism from $(\F{\underline{\implObj}}, \implCoalg)$ to $(\F{1}, \specCoalg)$ consists of a map $\Phi : \underline{\implObj} \to \omega$ and a 2-cell \[ \Phi \semi \specCoalg \ge \implCoalg \semi{} \sigFun\Phi \]
	demonstrating that $\Phi$ satisfies the lax generalized amortization condition.
	In this case, the inequality condition exactly requires that
	\[ \Phi(d) + \costof{\specCoalg} \ge \costof{\implCoalg}(d) + \Phi(\behof{\implCoalg}(d)), \]
	which matches the traditional amortization condition
	\[ \costof{\specCoalg} \ge \costof{\implCoalg}(d) + \Phi(\behof{\implCoalg}(d)) - \Phi(d) \]
  when the cost model is commutative and has additive inverses.
\end{example}

\begin{remark}
	$\sigFun$-coalgebras and their colax morphisms form a bicategory, where a 2-cell from $(\Phi, \varphi)$ to $(\Phi', \varphi')$ is a 2-cell $\Phi \Rightarrow \Phi'$ in $\catC$ along with a coherence condition on $\varphi$ and $\varphi'$~\cite{lack>2005}.
	In the restricted case of 2-posets, the coherence condition is trivialized.
	Then, a 2-cell $\Phi \leq \Phi'$ justifies that $\Phi$ expresses the amortization argument with at most as much overhead as $\Phi'$.
	For instance, in Example~\ref{ex:colax}, if $\Phi$ is a colax morphism, then so is $\Phi'(d) = \Phi(d) + \cost{c}$ for any $c : \omega$, using the commutativity of addition in $\omega$.
	Then, $\Phi \leq \Phi'$, since $\Phi'$ unnecessarily adds $\cost{c}$ cost.
\end{remark}

The advantages and generalizations from the 1-categorical development translate immediately to the 2-categorical setting.
For example, the specification carrier can be altered to expose some state, and the monad can be varied to amortize other effects.

\begin{example}\label{ex:stack}
  Extending Example~\ref{ex:dynamic-array}, we can treat a dynamically-resizing array as a stack by adding a method to pop the most-recently-added element.
  Consider the same signature $\sigFun$ as in Example~\ref{ex:queue}; we provide coalgebra implementations that represent stacks instead of queues:
  \begin{enumerate}
    \item
      In order to implement a $\sigFun$-coalgebra, we must not only keep track of the length of the stack, but all the elements stored in the stack, since the pop method produces an element of type $E$.
      The specification uses state type $\specObj = \F*{\listty{E}}$, storing a specification-level list of elements.
      Its transition morphism $\specCoalg$ is defined as follows:
      \begin{center}
        \iblock{
          \mrow{\specCoalg : \listty{E} \pto \sigFun(\listty{E})}
          \mrow{\specCoalg~\proj{push}~l~e = \charge{3}{\ret{\consex{e}{l}}}}
          \mrow{\specCoalg~\proj{pop}~\nilex = \inj{1}(\ret{\triv})}
          \mrow{\specCoalg~\proj{pop}~(\consex{e}{l}) = \charge{2}{\inj{2}(e, \ret{l})}}
        }
      \end{center}
      The push method behaves as in Example~\ref{ex:dynamic-array}, and the new pop method charges $\cost{2}$ for a pop on a nonempty array, terminating if the array is empty.

    \item
      The data structure implementation uses a similar state type,
      \[ \underline{\implObj} = \sum_{n : \Nat} \mathsf{array}_E[2^n - 1, 2^{n+2}-1). \]
      The only difference is a looser bound on the range of the array length.
      The push method stays the same, and the new pop method is similar, where $\mathsf{init}$ and $\mathsf{last}$ get the initial segment and last element of an array, respectively.
      \begin{center}
        \iblock{
          \mrow{\implCoalg : \underline{\implObj} \to E \to \F{\implObj}}
          \mrow{\implCoalg~\proj{push}~(n, a)~e \mid (\length{a} + 1 \equiv 2^{n + 2} - 1) = \charge*{3 + \length{a}}{\ret{\suc{n}, \consright{a}{e}}}}
          \mrow{\implCoalg~\proj{push}~(n, a)~e \mid \text{otherwise} = \charge{1}{\ret{n, \consright{a}{e}}}}
          \mrow{\implCoalg~\proj{pop}~(0, \nilex) = \inj{1}(\ret{\triv})}
          \mrow{\implCoalg~\proj{pop}~(\suc{n}, a) \mid (\length{a} \equiv 2^n - 1) = \charge*{2 + (\length{a} - 1)}{\inj{2}{(\mathsf{last}(a), \ret{n, \mathsf{init}(a)})}}}
          \mrow{\implCoalg~\proj{pop}~(\suc{n}, a) \mid \text{otherwise} = \charge{1}{\inj{2}{(\mathsf{last}(a), \ret{\suc{n}, \mathsf{init}(a)})}}}
        }

      \end{center}
  \end{enumerate}
  Let $c(x, y) = \max(2 \cdot (x - y), y - x)$.
  The program $\Phi(n, a) = \charge*{c(\length{a}, 2^{n+1}-1)}{\ret{\mathsf{toList}(a)}}$ is a coalgebra morphism representing the amortized analysis, containing the traditional cost-level potential function alongside a behavioral simulation that converts the size-bounded array to a specification-level list.
  Intuitively, an array is in a lower potential state the closer the number of elements it stores is to the middle of the bounds, since a resize is necessarily many operations away; this is mathematically justified by the potential function.
  When moving away from this midpoint, the potential function meets the amortization condition.
  However, when moving towards this midpoint, the potential is decreasing; in this case, we do not need all the cost provided by the specification, and the implementation only meets the specification laxly.
  Thus, the morphism $\Phi$ is not a strict coalgebra morphism, but it is a colax coalgebra morphism, guaranteeing that the amortized implementation is at least as efficient as the specification suggests.
\end{example}

\begin{example}\label{ex:non-free}
  By choosing a non-free carrier, we can ``lazily'' avoid some computation if its results are never needed while amortizing the cost if the results are eventually needed.
  A unary counter can be represented by the signature
  \[ \sigFun \X = \X \times (\Nat \rtimes \X) \times \F{1}, \]
  where the operations increment the counter, compute the total, and terminate the counter, respectively.
  As our cost model, we charge $\cost{1}$ per successor.

  \begin{enumerate}
    \item
      The specification uses carrier $\specObj = \F{\Nat}$, storing a single natural number.
      \begin{center}
        \iblock{
          \mrow{\specCoalg : \F{\Nat} \to \sigFun(\F{\Nat})}
          \mrow{\specCoalg~\proj{increment}~s = \bindex{s}{n}\charge{1}{\ret{\suc{n}}}}
          \mrow{\specCoalg~\proj{total}~s = \bindex{s}{n}(n, \ret{n})}
          \mrow{\specCoalg~\proj{terminate}~s = \bindex{s}{n}{\ret{\triv}}}
        }
      \end{center}
      Note that we explicitly write $\F{\Nat} \to \sigFun(\F{\Nat})$ instead of our usual shorthand $\Nat \pto \sigFun(\F{\Nat})$ to emphasize similarity with the following implementation.

    \item
      The data structure implementation has carrier $\implObj$ as the following comma object, a subobject of the lazy pair $\F{\Nat} \times \F{1}$ where the cost of the $\F{1}$ component is at most the cost of the $\F{\Nat}$ component:
      \[\begin{tikzcd}
        D & {\F{\Nat}} \\
        {\F{1}} & {\F{1}}
        \arrow[from=1-1, to=1-2]
        \arrow[from=1-1, to=2-1]
        \arrow["\le"{description}, from=1-2, to=2-1, draw=none]
        \arrow["{\F{\triv}}", from=1-2, to=2-2]
        \arrow[r,-,double equal sign distance,double, from=2-1, to=2-2]
      \end{tikzcd}\]
      The implementation stores this lazy pair of $\F{\Nat}$ and $\F{1}$ computations to ``hedge its bets'':
      in case the final total is never used, the computation of type $\F{1}$ will waste the cost from the unused successors.
      \begin{center}
        \iblock{
          \mrow{\implCoalg : D \to \sigFun(D)}
          \mrow{\implCoalg~\proj{increment}~d = ((\bindex{d~\proj{proj}_1}{n}\charge{1}{\ret{\suc{n}}}), (d~\proj{proj}_2))}
          \mrow{\implCoalg~\proj{total}~d = \bindex{d~\proj{proj}_1}{n}(n, (\ret{n}, \ret{\triv}))}
          \mrow{\implCoalg~\proj{terminate}~d = d~\proj{proj}_2}
        }
      \end{center}
      The increment operation keeps the $\F{\Nat}$ and $\F{1}$ projections independent, accumulating one cost in the $\F{\Nat}$ component while leaving the $\F{1}$ component as-is.
      However, to implement the totaling operation, the first projection must be used to compute the eager $\Nat$ before the lazy ``hedging'' can take place.
      By the algebra structure on a product, all cost stored in the $\F{\Nat}$ component will flow to \emph{both} components going forward: since we needed the total after all, even the amortized cost of the second component of type $\F{1}$ must be increased.
      In the termination operation, though, the opposite situation occurs: we don't need the first component after all, so by taking the second projection, we may save on cost incurred by increments since the last total was computed.
  \end{enumerate}
  The first projection $\Phi : D \to \F{\Nat}$ is a colax morphism from $(\implObj, \implCoalg)$ to $(\specObj, \specCoalg)$, where the lax amortization condition in the termination case is justified by the cost inequality guaranteed by the comma object.
\end{example}

In the remainder of the paper, we will work with 1-categories for simplicity, although the constructions readily generalize to the 2-categorical setting.
\section{Splitting and Combining Potential}\label{sec:profunctor}

In more complex amortized analyses, the behavior of an amortized data structure can branch, breaking the data structure into multiple parts or combining multiple instances.
As discussed by Okasaki~\cite[\S 5.3]{okasaki>1999}, in this scenario, the amortization condition should consider the sum of the potentials of the input and output states.
Informally, we say
\[ \costof{\specCoalg} = \costof{\implCoalg}(\textit{Input}) + \sum_{d' \in \behof{\implCoalg}(\textit{Input})} \Phi(d') - \sum_{d \in \textit{Input}} \Phi(d), \]
where $\textit{Input}$ is a set of input states and $\behof{\implCoalg}(\textit{Input})$ is the corresponding set of output states.
In the case that there is a single input and a single output, this condition is equivalent to the usual condition.

To make sense this in our presentation, we consider two additional structures.
First, we represent multiple data structure states in parallel via a monoidal product on $\catC$, recovering and formalizing the idea of summing the potentials of states for the typical case of free algebra carriers.
Then, we generalize signatures from endofunctors to endo\emph{profunctors} to allow for multiple parallel inputs.

\subsection{Parallel States with Additive Potential}

To represent multiple simultaneous output states, we require that $\catC$ come equipped with a symmetric monoidal structure $(\top, \otimes)$.
Then, for example, a method that splits the amortized data structure into two parts will be represented by the signature by $\sigFun \X = \X \otimes \X$.
Such a symmetric monoidal structure $(\top, \otimes)$ exists in the category of algebras $\catC = \Alg{T}$ when the monad $T$ is strong and commutative.
The writer monad $\Writer{(-)}$ is commutative exactly when the monoid on $\Cost$ is commutative; this is often a reasonable assumption in the setting of cost analysis.
When $T$ is commutative, the adjunction $\F{} \dashv \U{} : \Alg{T} \to \Set$ is also lax monoidal, which implies that $\F{}$ is a strong monoidal functor:
\begin{align*}
  \top &\cong \F{1} \\
  \F{\A} \otimes \F{\B} &\cong \F*{\A \times \B}
\end{align*}
The forward direction of the first isomorphism adds together the potential stored in the components.
Thus, we can support the branching generalization of amortized analysis via the monoidal product.

\begin{example}
  Let $\sigFun \X = \X \otimes \X$, representing a single method that splits a data structure into components.
  Suppose $(\F{\underline{\implObj}}, \implCoalg)$ and $(\F{1}, \specCoalg)$ are $\sigFun$-coalgebras.
  To prove that the former is an amortized implementation of a latter, we give a potential function $\Phi : \underline{\implObj} \to \Cost$ such that the following square commutes:
  \[\begin{tikzcd}
    {\F{\underline{\implObj}}} & {\F{\underline{\implObj}} \otimes \F{\underline{\implObj}}} \\
    {\F{1}} & {\F{1} \otimes \F{1}}
    \arrow["\Phi", from=1-1, to=2-1]
    \arrow["{\Phi \otimes \Phi}", from=1-2, to=2-2]
    \arrow["\implCoalg", from=1-1, to=1-2]
    \arrow["\specCoalg", from=2-1, to=2-2]
  \end{tikzcd}\]
  Since $\F{1} \otimes \F{1} \cong \F{1}$, the behavioral component of both paths are trivial.
  The condition on costs can be stated as follows, using the fact that the map $\F{1} \otimes \F{1} \to \F{1}$ adds costs:
  \[ \Phi(d) + \costof{\specCoalg} = \costof{\implCoalg}(d) + (\Phi(\behof{\implCoalg}(d)_1) + \Phi(\behof{\implCoalg}(d)_2)) \]
  Equivalently, we may write:
  \[ \Phi(d) + \costof{\specCoalg} = \costof{\implCoalg}(d) + \sum_{i \in \set{1, 2}} \Phi(\behof{\implCoalg}(d)_i) \]
  Returning to the informal amortization condition, this makes precise the notion of having multiple states output states, using the monoidal product to capture the addition of potentials.
\end{example}

\subsection{Algebraic and Coalgebraic Operations via Profunctors}

Some data structures support ``algebraic'' operations that, for example, may combine multiple instances.
Such situations can be described by the use of \emph{$\sigFun$-algebras} rather than coalgebras and their (colax) morphisms.
However, in general, an operation may have many inputs and outputs.
For example, an operation may take two instances of a data structure and produce two updated instances, $(\X \otimes \X) \to (\X \otimes \X)$.
As defined, a coalgebra for an endofunctor $\sigFun$ consists of a carrier $\implObj$ and a transition morphism $\implCoalg : \implObj \to \sigFun \implObj$ that takes a single state and provides possibilities for transition.
To support such ``binary methods'' (as in object-oriented programming \cite{tews>2001}) with multiple inputs while still retaining coalgebraic operations with multiple outputs, we generalize to coalgebras for an endoprofunctor $\sigFun$, which will provide the possibility for multiple input and output states by supporting both covariant and contravariant uses of the carrier.

When $(\catC, \top, \otimes)$ is additionally equipped with a closure ${\multimap} : \op{\catC} \times \catC \to \catC$,
one may naively attempt to use to support a method with two inputs (and two outputs) using signature
$\sigFun \X = \X \multimap (\X \otimes \X)$,
so that a $\sigFun$-coalgebra $\implObj \to \sigFun \implObj$ is equivalent to a map $\implObj \otimes \implObj \to \implObj \otimes \implObj$.
However, this definition of $\sigFun$ is not functorial: $\X$ is used in a contravariant position.
To address this, we generalize from endofunctors to endoprofunctors, analogous to generalizing from functions to relations or from coinductive types to existential types.
A profunctor $\catC \pro \catD$ is a functor $\op{\catD} \times \catC \to \Set$.
To represent an arbitrary signature for an abstract data type, we will use an endoprofunctor $\sigFun : \catC \pro \catC$ in place of an endofunctor, taking in both a covariant and a contravariant copies of what will ultimately be $\X$.
For example, if we let \[ \sigFun(\X^-, \X^+) = (\X^- \otimes \X^-) \multimap (\X^+ \otimes \X^+), \]
we describe a signature that takes in a pair of states and produces a new pair of states.
We now recall the definition of a coalgebra for an endoprofunctor, using the bicategorical generalization coalgebra.

\begin{definition}\label{def:pro-coalgebra}
  Let $\sigFun$ be an endoprofunctor.
  A $\sigFun$-coalgebra is a pair $(\implObj, \implCoalg)$ of a carrier object $\implObj : 1 \pro \catC$ and a transition morphism $\implCoalg : \implObj \to \sigFun \circ \implObj$~\cite{nlab:algebra_for_a_profunctor}.
\end{definition}
Note that profunctors $1 \pro \catC$ are equivalent to presheaves $\Presh{\catC}$.
For cost analysis, we will continue to use $\catC \isdef \Alg{\Writer{(-)}}$.
Coalgebras for endoprofunctors encompass coalgebras for endofunctors.

\begin{theorem}
  Let $\sigFun : \catC \to \catC$ be an arbitrary endofunctor, and define $\widetilde{\sigFun} : \catC \pro \catC$ by:
  \[ \widetilde{\sigFun}(\X^-, \X^+) = \X^- \multimap \sigFun \X^+ \]
  Then, a $\sigFun$-coalgebra with carrier $\X$ is equivalent to a $\widetilde{\sigFun}$-coalgebra with carrier $\Yoneda{\X}$.
\end{theorem}
\begin{proof}
  By the Yoneda lemma, maps $\implCoalg : \Yoneda{\implObj} \to \widetilde{\sigFun} \circ \Yoneda{\implObj}$ are equivalent to elements of $\widetilde{\sigFun}(\implObj, \implObj)$.
  By definition, we have $\widetilde{\sigFun}(\implObj, \implObj) = \implObj \multimap \sigFun \implObj$, precisely the definition of a $\sigFun$-coalgebra transition morphism.
\end{proof}

While coalgebras over endofunctors are definable as coalgebras over endoprofunctors, this new environment allows more flexibility in the contravariant position.
Morphisms of coalgebras over an endoprofunctor generalize morphisms of coalgebras over an endofunctor, as well.

\begin{definition}
  Let $(\implObj, \implCoalg)$ and $(\specObj, \specCoalg)$ be $\sigFun$-coalgebras, where $\sigFun$ is an endoprofunctor.
  A morphism of $\sigFun$-coalgebras from $(\implObj, \implCoalg)$ to $(\specObj, \specCoalg)$ consists of a morphism $\Phi : \implObj \to \specObj$ that preserves the $\sigFun$-coalgebra structure, as before:
  \[\begin{tikzcd}
    \implObj & \sigFun \circ \implObj \\
    \specObj & \sigFun \circ \specObj
    \arrow["\implCoalg", from=1-1, to=1-2]
    \arrow["\specCoalg", from=2-1, to=2-2]
    \arrow["\Phi", from=1-1, to=2-1]
    \arrow["\sigFun \circ \Phi", from=1-2, to=2-2]
  \end{tikzcd}\]
  Note that $\implObj$ and $\specObj$ are presheaves in $\Presh{\catC}$ and $\Phi$ is a morphism of presheaves.
\end{definition}

When $\implObj = \Yoneda{\implObj_0}$ and $\specObj = \Yoneda{\specObj_0}$, the situation becomes simpler.
Since the Yoneda embedding is fully faithful, it is equivalent to give a map $\Phi : \implObj_0 \to \specObj_0$.
By the Yoneda lemma, the structure preservation requirement can be simplified to the following:
\[\begin{tikzcd}
  1 & \sigFun(\implObj_0, \implObj_0) \\
  \sigFun(\specObj_0, \specObj_0) & \sigFun(\implObj_0, \specObj_0)
  \arrow["\implCoalg", from=1-1, to=1-2]
  \arrow["{\sigFun(\Phi, \specObj_0)}", from=2-1, to=2-2]
  \arrow["\specCoalg", from=1-1, to=2-1]
  \arrow["{\sigFun(\implObj_0, \Phi)}", from=1-2, to=2-2]
\end{tikzcd}\]
Viewing the covariant and contravariant positions as outputs and inputs, respectively, this formalizes of the amortization condition with multiple inputs and outputs.
We demonstrate this via a concrete example.

\begin{example}
  Define endoprofunctor $\sigFun(\X^-, \X^+) = (\X^- \otimes \X^-) \multimap (\X^+ \otimes \X^+)$, representing a method that takes two states of a data structure and provides two new states.
  Suppose $(\Yoneda*{\F{\underline{\implObj}}}, \implCoalg : \sigFun(\F{\underline{\implObj}}, \F{\underline{\implObj}}))$ and $(\Yoneda*{\F{1}}, \specCoalg : \sigFun(\F{1}, \F{1}))$ are $\sigFun$-coalgebras.
  Then, the above condition simplifies to:
  \[ (\Phi(d_1) + \Phi(d_2)) + \costof{\specCoalg} = \costof{\implCoalg}(d) + (\Phi(\behof{\implCoalg}(d)_1) + \Phi(\behof{\implCoalg}(d)_2)) \]
  where $d = (d_1, d_2) : \underline{\implObj} \times \underline{\implObj}$.
  Equivalently, we may write:
  \[ \sum_{i \in \set{1, 2}} \Phi(d_i) + \costof{\specCoalg} = \costof{\implCoalg}(d) + \sum_{i \in \set{1, 2}} \Phi(\behof{\implCoalg}(d)_i) \]
  Up to our usual movement of the starting potential across the equation, this is precisely the generalized amortization equation given above, using the structure $\sigFun$ to make precise the informal $\textit{Input}$ notion.
\end{example}

Using a monoidal product, we allowed multiple states to exist simultaneously, adding any potential they contain.
Then, using profunctors, we generalized the notion of signature to support arbitrary contravariant data, allowing arbitrarily many inputs and outputs to a method.

\section{Composition of Amortized Data Structures}\label{sec:composition}

Thus far, we have considered coalgebra morphisms in isolation, each showing that one data structure implementation matches a specification up to amortization.
Using the fact that $\sigFun$-coalgebras and coalgebra morphisms form a category, we may compose potential functions to support different levels of amortized abstraction.

\begin{example}
  Recall Examples~\ref{ex:allocation}~and~\ref{ex:allocation-potential}, where the specification $(\specObj, \specCoalg)$ purports to incur $\cost{1}$ of cost every operation while really the implementation $(\implObj, \implCoalg)$ incurs $\cost{8}$ of cost every eight operations.
  Alternatively, we may view $(\implObj, \implCoalg)$ as the specification, which can be implemented by an even more amortized scheme.
  For example, we may define a coalgebra $(\implObj', \implCoalg')$ that incurs $\cost{16}$ of cost every sixteen operations, which is an amortized implementation of $(\implObj, \implCoalg)$ via a new potential function $\Phi' : (\implObj', \implCoalg') \to (\implObj, \implCoalg)$.
  Of course, we may then compose these potential functions,
  \[\begin{tikzcd}
    {(\implObj', \implCoalg')} & {(\implObj, \implCoalg)} & {(\specObj, \specCoalg)},
    \arrow["{\Phi'}", from=1-1, to=1-2]
    \arrow["\Phi", from=1-2, to=1-3]
  \end{tikzcd}\]
  showing that $(\implObj', \implCoalg')$ is an implementation of the original specification $(\specObj, \specCoalg)$ after all.
\end{example}

More commonly, we may wish to compose coalgebras with different signatures, using one amortized data structure to implement another.
For example, we may wish to use a pair of stacks, with an amortized implementation in terms of arrays in Example~\ref{ex:stack}, to implement a queue, as given in Example~\ref{ex:queue}.
The coalgebras and morphisms for pairs of stacks and queues exist in different categories, $\Coalg{\sigFun_\text{s} \times \sigFun_\text{s}}$ and $\Coalg{\sigFun_\text{q}}$, where both $\sigFun_\text{s}$ and $\sigFun_\text{q}$ are (coincidentally!) both the signature from Example~\ref{ex:stack} extended with a method $\Nat \rtimes \X$ for computing the number of elements stored in the data structure.

Rather than viewing each category of coalgebras $\Coalg{\sigFun}$ in isolation for a given signature $\sigFun$, we may think of $\Coalg{-} : \Fun{\Alg{T}}{\Alg{T}} \to \Cat$ as an indexed category whose (covariant) Grothendieck construction
$\int^{\sigFun} \Coalg{\sigFun}$
is the category with objects $(\sigFun_\implObj, (\implObj, \implCoalg))$, where $(\implObj, \implCoalg)$ is a $\sigFun_\implObj$-coalgebra.
In other words, an object consists of a signature $\sigFun_\implObj$ along with an implementation of that signature.
A morphism from $(\sigFun_\implObj, (\implObj, \implCoalg))$ to $(\sigFun_\specObj, (\specObj, \specCoalg))$ consists of a natural transformation $\varphi : \sigFun_\implObj \to \sigFun_\specObj$ implementing the operations of $\sigFun_\specObj$ in terms of the operations of $\sigFun_\implObj$, along with a $\sigFun_\specObj$-coalgebra morphism from $(\implObj, \implObj \xrightarrow{\implCoalg} \sigFun_\implObj\implObj \xrightarrow{\varphi \implObj} \sigFun_\specObj\implObj)$ to $(\specObj, \specCoalg)$ performing an amortized analysis on the $\varphi$-translated implementation.

\begin{example}
  Let $\sigFun_\text{s}$ be the signature functor from Example~\ref{ex:stack} for stacks, and let $\sigFun_\text{c} \X = \X \times (\F{1} + \X)$ be the signature of for a counter with successor and predecessor methods.
  Let $\implObj$ be the carrier used to implement stacks.
  The amortized analysis of stacks induces a morphism
  \[ (\sigFun_\text{s}, (\implObj, \implCoalg)) \to (\sigFun_\text{s}, (\F*{\listty{E}}, \specCoalg)) \]
  whose translation component $\sigFun_\text{s} \to \sigFun_\text{s}$ is the identity.
  Separately, let $(\F{\Nat}, \specCoalg_\text{counter})$ be a $\sigFun_\text{c}$-coalgebra specification similar to the $\sigFun_\text{s}$ coalgebra $(\F*{\listty{E}}, \specCoalg)$, keeping the same costs but only storing the number of elements stored rather than the elements themselves.
  We can give a morphism to $(\F{\Nat}, \specCoalg_\text{counter})$
  \[ (\sigFun_\text{s}, (\F*{\listty{E}}, \specCoalg)) \to (\sigFun_\text{c}, (\F{\Nat}, \specCoalg_\text{counter})) \]
  that implements counters in terms of stacks of a sentinel value $e_0 : E$.
  Composing these morphisms, we get a map
  \[ (\sigFun_\text{s}, (\implObj, \implCoalg)) \to (\sigFun_\text{c}, (\F{\Nat}, \specCoalg_\text{counter})) \]
  that implements natural number counters in terms of a dynamically resizing array, mixing signatures.
\end{example}

To implement an amortized queue as a pair of amortized stacks, we must first be able to pair two coalgebras.
Fortunately, when restricted to signature functors equipped with a tensorial strength (which includes all signatures considered in this work), the indexed category $\Coalg{-}$ is lax monoidal.
\begin{theorem}
  Let $\catC$ be a symmetric monoidal category.
  Then, $\Coalg{-} : \Strong{\catC}{\catC} \to \Cat$ is a lax monoidal pseudofunctor, lifting the symmetric monoidal structure $(\top, \otimes)$ to coalgebras.
\end{theorem}
\begin{proof}
  For coalgebras $(\implObj_1, \implCoalg_1) : \Coalg{\sigFun_1}$ and $(\implObj_2, \implCoalg_2) : \Coalg{\sigFun_2}$, we define $(\implObj_1, \implCoalg_1) \otimes (\implObj_2, \implCoalg_2)$ to have carrier $\implObj \isdef \implObj_1 \otimes \implObj_2$ and transition morphism
\[\begin{tikzcd}
	{\implObj} & {\implObj \times \implObj} &&& {(\sigFun_1 \implObj_1 \otimes \implObj_2) \times (\implObj_1 \otimes \sigFun_2 \implObj_2)} & {\sigFun_1 \implObj \times \sigFun_2 \implObj}
	\arrow["\Delta", from=1-1, to=1-2]
	\arrow["{(\implCoalg_1 \otimes \implObj_2) \times (\implObj_1 \otimes \implCoalg_2)}", from=1-2, to=1-5]
	\arrow["{t_1 \times t_2}", from=1-5, to=1-6]
\end{tikzcd}\]
  where $t_1$ and $t_2$ are strengths for $\sigFun_1$ and $\sigFun_2$, respectively.
  The unit is $(\top, \top \xrightarrow{\triv} 1) : \Coalg{1}$.
\end{proof}

If $(\specObj_\text{s}, \specCoalg_\text{s}) : \Coalg{\sigFun_\text{s}}$ is the specification coalgebra for a stack,
then $(\specObj_\text{s}, \specCoalg_\text{s}) \otimes (\specObj_\text{s}, \specCoalg_\text{s}) : \Coalg{\sigFun_\text{s} \times \sigFun_\text{s}}$
is the compound specification for a pair of stacks.
We may then hope to give a morphism
\[ (\sigFun_\text{s} \times \sigFun_\text{s}, (\specObj_\text{s}, \specCoalg_\text{s}) \otimes (\specObj_\text{s}, \specCoalg_\text{s})) \to (\sigFun_\text{q}, (\specObj_\text{q}, \specCoalg_\text{q})) \]
implementing amortized queues in terms of a pair of stacks.
The translation morphism $\sigFun_\text{s} \times \sigFun_\text{s} \to \sigFun_\text{q}$ would have to describe each queue operation in terms of a stack operation.
However, implementing a queue operation may require more than one stack operation, popping all the elements of the ``inbox'' stack in the case the ``outbox'' stack is empty.
To allow each queue operation to perform arbitrarily many stack operations, we instead treat $\Coalg{-}$ as category indexed in the coKleisli category of the cofree comonad comonad, $\Cofree{(-)} : \Fun{\Alg{T}}{\Alg{T}} \to \Fun{\Alg{T}}{\Alg{T}}$.
This approach is the formal dual of recent work by Grodin~and~Spivak~\cite{grodin-spivak>2024-poly-morphic-effect-handlers} about algebraic effect handlers; in this sense, a translation morphism can be viewed as a ``coalgebraic coeffect cohandler''.

\begin{definition}
  The pseudofunctorial action of the indexed category $\Coalg{-} : \cat{coKl}(\Cofree{(-)}) \to \Cat$ is given by post-composition of coalgebra maps $\implCoalg : \implObj \to \Cofree{\sigFun_\implObj}\implObj$ with the coKleisli extension of a translation morphism $\varphi : \Cofree{\sigFun_\implObj} \to \sigFun_\specObj$, written $\varphi^\dagger : \Cofree{\sigFun_\implObj} \to \Cofree{\sigFun_\specObj}$.
  \begin{align*}
    &\Coalg[0]{\sigFun} = \CCoalg{\Cofree{\sigFun}} \\
    &\Coalg[1]{\varphi}_0(\implObj, \implObj \xrightarrow{\implCoalg} \Cofree{\sigFun_\implObj}\implObj) = (\implObj, \implObj \xrightarrow{\implCoalg} \Cofree{\sigFun_\implObj}\implObj \xrightarrow{\varphi^\dagger \implObj} \Cofree{\sigFun_\specObj}\implObj) \\
    &\Coalg[1]{\varphi}_1((\implObj, \implCoalg) \xrightarrow{\Phi} (\specObj, \specCoalg)) = (\implObj, (\implCoalg \semi{} \varphi^\dagger \implObj)) \xrightarrow{\Phi} (\specObj, (\specCoalg \semi{} \varphi^\dagger \specObj))
  \end{align*}
  We write $\CCoalg{\Cofree{\sigFun}}$ for the category of comonad coalgebras over the cofree comonad $\Cofree{\sigFun}$ to simplify composition.
  However, this category is equivalent to $\Coalg{\sigFun}$, so the object part is the same as when indexed over the category of endofunctors.
\end{definition}

In the fibered category
$\int^{\sigFun : \cat{coKl}(\Cofree{(-)})} \Coalg{\sigFun}$,
the objects are the same as before, but the signature translation of a morphism from $(\sigFun_\implObj, (\implObj, \implCoalg))$ to $(\sigFun_\specObj, (\specObj, \specCoalg))$ now uses the cofree comonad on $\sigFun_\implObj$ in its domain, $\varphi : \Cofree{\sigFun_\implObj} \to \sigFun_\specObj$, translating the operations of $\sigFun_\specObj$ to finitely many $\sigFun_\implObj$ operations.

\begin{example}
  In this category, we may implement a morphism
  \[ (\sigFun_\text{s} \times \sigFun_\text{s}, (\specObj_\text{s}, \specCoalg_\text{s}) \otimes (\specObj_\text{s}, \specCoalg_\text{s})) \to (\sigFun_\text{q}, (\specObj_\text{q}, \specCoalg_\text{q})), \]
  using amortized stacks to implement amortized queues.
  The translation morphism $\sigFun_\text{s} \times \sigFun_\text{s} \to \sigFun_\text{q}$ is analogous to the behavior of the implementation coalgebra in Example~\ref{ex:queue} but generic in stack implementation, and the coalgebra morphism is precisely as in Example~\ref{ex:queue}.
  Pre-composing this with the morphism from Example~\ref{ex:stack}, we implement the queue specification via a pair of stacks up to amortization, which in turn are implemented as arrays up to amortization.
\end{example}

The techniques considered here generalize to bicategories (Section~\ref{sec:lax}) and mixed-variance signatures given by profunctors (Section~\ref{sec:profunctor}).
\section{Conclusion}\label{sec:conclusion}

In this work, we observed that the condition imposed on a potential function in amortized analysis is generalized by the commutativity condition for a coalgebra in the category of writer monad algebras.
Using this perspective, we gave simple, clear accounts of examples, and we generalized amortized analysis to other effectful settings.
We expanded this definition to include branching amortization given a commutative cost model using profunctors, and we investigated the composition of amortized data structures via the indexed category of coalgebras.
Finally, we observed that this development makes sense in a bicategory, where colax coalgebra morphisms give a more relaxed and practically useful definition of amortization.

The key ideas and some examples presented have been mechanized inside Calf~\cite{niu-sterling-grodin-harper>2022,grodin-niu-sterling-harper>2024} in the Agda proof assistant.
Basic definitions, such as functor, coalgebra, and (colax) coalgebra morphism, are formulated explicitly, with the colax commutativity condition expressed using the program inequality of Grodin~\etal~\cite{grodin-niu-sterling-harper>2024}.
The analyses themselves are structured analogous to the on-paper reasoning presented here, with the added formality required for mechanized proof.

\subsection{Related Work}

This development principally expands upon the early work of Grodin~and~Harper~\cite{grodin-harper>2023}, which first proposed to use coalgebraic techniques to reason about amortized analysis.
Amortized analysis was first characterized by Tarjan~\cite{tarjan>1985}, providing a technique for describing the cost of data structure operations that takes into account the evolving state of the data structure and considers cost in aggregate, in contrast with algorithmic worst-case upper bounds on a per-operation basis.
Representing abstract data types as coalgebras and simulations as coalgebra morphisms has been well studied; see, \eg, Jacobs~\cite{jacobs>1996-objects-coalgebraically,jacobs>2017} for an overview.
Formalization of cost analysis using a writer monad has been studied extensively; see Niu~\etal~\cite{niu-sterling-grodin-harper>2022} for a comprehensive review of the literature.
Mechanizations of amortized cost analysis has also been pursued, but in the traditional algebraic form~\cite{plotkin-power>2008,nipkow-brinkop>2019,niu-sterling-grodin-harper>2022}.
Type theories for automatic cost inference, such as AARA~\cite{hofmann-jost>2003,hoffmann-jost>2022}, are based on the potential functions of the physicist's method of amortized analysis.

\subsection{Future Work}

While the cost algebra carriers of most coalgebras considered are free (\ie, of the form $\F{\A}$), in Example~\ref{ex:non-free} we use a product of free algebras to lazily avoid some computation.
Although this example is somewhat contrived and the laziness does not affect the asymptotic complexity of operations, we leave it as future work to determine if such methods could be used to save cost in a more realistic data structure.
Similar observations connecting amortization to laziness have been made by Okasaki~\cite[\S 6]{okasaki>1999}; however, his work uses memoized suspensions to amortize in a call-by-value setting, whereas here we use linearity at the level of computation types in call-by-push-value to address this issue.
Future work may make this connection between amortization, laziness, and linearity more precise.

When using endoprofunctors as signatures, though, the carriers of all coalgebras considered are representable (\ie, of the form $\Yoneda{\X}$); we hope to investigate in future work whether non-representable carriers would be of use for amortized analysis.

In the pioneering works on amortized analysis, Sleator and Tarjan~\cite{tarjan>1985,sleator-tarjan>1985-paging} describe the ``move-to-front'' paging algorithm and prove that it is competitive relative to \emph{any} other algorithm.
This argument, while also using amortization, does not immediately appear to fit into our coalgebraic framework.
The potential is computed by comparing the move-to-front algorithm to another algorithm; we conjecture that the potential is behaving like a generalized bisimulation relation rather than a morphism.
While we view amortized analysis as a directed simulation here (\ie, a morphism of coalgebras), we suspect that a double categorical approach may be valuable for understanding a symmetric generalization more akin to bisimulation.
Cospans of coalgebras have been used to generalize bisimulations \cite{staton>2011}, which can serve as vertical morphisms to add to the bicategory of colax coalgebras.
Additionally, Goncharov~\etal~\cite{goncharov-hofmann-nora-schroder-wild>2022} propose a double category of (co)lax coalgebras of a lax-double functor, generalizing from a fundamentally double categorical perspective.

In the present work, we consider coalgebras in a category of algebras.
This construction appears more symmetrically in the development of mathematical operational semantics, there called a bialgebra, representing a transition system in the style of operational semantics.
We leave it to future investigation to understand the connections between amortized analysis and operational semantics.

In our formalization within Calf, the notion of coalgebra morphism is defined explicitly, in contrast with the notion of cost algebra that is built in as a primitive.
We hope to incorporate coalgebras and coalgebra morphisms more fundamentally, allowing amortized analysis to be more seamless within the language rather than embedded.

\section*{Acknowledgements}
The authors wish to thank Max New, Yue Niu, and David Spivak for insightful conversations that directly impacted the direction of this work and Zachary Battleman, Runming Li, Parth Shastri, and Jonathan Sterling for fruitful adjacent collaboration that contributed broader inspiration.
Additionally, the authors thank the anonymous thank the anonymous reviewers for their thoughtful comments.
This material is based upon work supported by the United States Air Force Office of Scientific Research under grant number FA9550-21-0009 (Tristan Nguyen, program manager) and the National Science Foundation under grant number CCF-1901381. Any opinions, findings and conclusions or recommendations expressed in this material are those of the authors and do not necessarily reflect the views of the AFOSR or NSF.
\bibliographystyle{./entics}
\bibliography{mfps,nlab}

\end{document}